\documentclass[10pt]{article}
\usepackage{caption}

\usepackage{cite}

\usepackage{subcaption}
\usepackage[utf8]{inputenc}

\usepackage{amssymb}
\usepackage{amsfonts}
\usepackage{amsmath}
\usepackage{graphicx}
\usepackage{indentfirst}
\usepackage{dsfont}

\usepackage{mathtools}

\usepackage{enumerate}
\usepackage{url}
\usepackage{lmodern}
\usepackage{newpxtext}

\usepackage{array}
\usepackage{multirow}
\usepackage{longtable}

\usepackage{amsthm}

\newtheorem{definition}{Definition}
\newtheorem{theorem}{Theorem}
\newtheorem{corollary}[theorem]{Corollary}
\newtheorem{proposition}[theorem]{Proposition}
\newtheorem{lemma}[theorem]{Lemma}

\usepackage[symbol]{footmisc}

\usepackage{tikz}
\definecolor{darkblue}{rgb}{0.15,0.35,0.55}
\definecolor{reddish}{rgb}{.8, 0.2, 0.2}
\usepackage[pdfpagelabels, linktocpage=true]{hyperref}
\hypersetup{
colorlinks=true,
citecolor=darkblue,
linkcolor=reddish,
urlcolor=darkblue,
pdfauthor={},
pdftitle={},
pdfsubject={}
}

\newcommand{\EX}{\mathbf{E}}

\long\def\ca#1\cb{} 

\newcommand{\becs}{\begin{cases}}
\newcommand{\bem}{\begin{matrix}}

\newcommand{\bsk}{\bigskip }

\newcommand{\dyad}[2]{|#1\rangle\langle#2|}

\newcommand{\encs}{\end{cases}}
\newcommand{\enm}{\end{matrix}}

\newcommand{\ket}[1]{|#1\rangle }

\newcommand{\mat}[1]{\left(\begin{matrix}#1\end{matrix}\right)}

\newcommand{\ot}{\otimes }

\newcommand{\Tr}{{\rm Tr}}

\newcommand{\AC}{{\mathcal A}}
\newcommand{\BC}{{\mathcal B}}

\newcommand{\DC}{{\mathcal D}}
\newcommand{\EC}{{\mathcal E}}

\newcommand{\HC}{{\mathcal H}}
\newcommand{\IC}{{\mathcal I}}

\newcommand{\LC}{{\mathcal L}}

\newcommand{\NC}{{\mathcal N}}

\newcommand{\XC}{{\mathcal X}}
\newcommand{\YC}{{\mathcal Y}}

\newcommand{\rB}{\textbf{r}}

\newcommand{\al}{\alpha }

\newcommand{\gm}{\gamma }

\newcommand{\Dl}{\Delta }

\newcommand{\Lm}{\Lambda }

\newcommand{\sg}{\sigma }

\DeclareMathOperator{\tr}{tr}

\def\outl#1{\par{\medskip\noindent\hspace*{0.1cm}\bf
      \mathversion{bold}#1\mathversion{normal}\smallskip} }
   \def\xa{} \def\xb{}  

 \def\outl#1{}\def\xa{}\def\xb{}

\ca
 \def\outl#1{\par{\medskip\noindent\hspace*{.5cm}\bf
      \mathversion{bold}#1\mathversion{normal}\smallskip} }
 \long\def\xa#1\xb{} 
  
\cb

\title{Queue-Channel Capacities with Generalized Amplitude Damping}

\author{Vikesh Siddhu, Avhishek Chatterjee, 
Krishna Jagannathan, Prabha Mandayam, Sridhar Tayur}
\begin{document}
\begin{center}
{\bf Queue-Channel Capacities with Generalized Amplitude Damping}\bsk\\
    \normalsize Vikesh Siddhu~$^1$,
    Avhishek Chatterjee~$^2$, Krishna Jagannathan~$^2$, Prabha Mandayam~$^3$, and Sridhar
    Tayur~$^4$\\
    \textit{$^1$ JILA, University of Colorado/NIST, 440 UCB, Boulder, CO 80309, USA}\\
    \textit{$^2$ Department of Electrical Engineering, Indian Institute of Technology Madras, Chennai 600036, India}\\
    \textit{$^3$ Department of Physics, Indian Institute of Technology Madras, Chennai 600036, India}\\
    \textit{$^4$ Quantum Computing Group, Tepper School of Business, Carnegie Mellon University, Pittsburgh PA 15213, USA}\\
Date: 28 Jul 2021
\vspace{.1cm}

\end{center}

\begin{abstract}
The generalized amplitude damping channel (GADC) is considered an important model for quantum communications, especially over optical networks. We make two salient contributions in this paper apropos of this channel. First, we consider a symmetric GAD channel characterized by the parameter $n=1/2,$ and derive its exact classical capacity, by constructing a specific induced classical channel. We show that the Holevo quantity for the GAD channel equals the Shannon capacity of the induced binary symmetric channel, establishing at once the capacity result and that the GAD channel capacity can be achieved without the use of entanglement at the encoder or joint measurements at the decoder. Second, motivated by the inevitable buffering of qubits in quantum networks, we consider a generalized amplitude damping \emph{queue-channel} ---that is, a setting where qubits suffer a waiting time dependent GAD noise as they wait in a buffer to be transmitted. This GAD queue channel is characterized by non-i.i.d. noise due to correlated waiting times of consecutive qubits. We exploit a conditional independence property in conjunction with additivity of the channel model, to obtain a capacity expression for the GAD queue channel in terms of the stationary waiting time in the queue. Our results provide useful insights towards designing practical quantum communication networks, and highlight the need to explicitly model the impact of buffering.
\end{abstract}

\section{Introduction}
\label{sec:intro}

There is considerable and growing interest in designing and setting up
large-scale quantum communication networks \cite{wehner2020}. To that end,
understanding the fundamental capacity limits of quantum communications in the
presence of noise is of practical importance. In this context, the inevitable
buffering of quantum states during communication tasks acts as an additional
source of decoherence. One concrete example of such buffering occurs at
intermediate nodes or quantum repeaters, where quantum states have to be stored
for a certain \emph{waiting time} until they are processed and transmitted
again~\cite{nemoto2016}. Indeed, while quantum states wait in buffer for
transmission, they continue to interact with the environment, and suffer a
\emph{waiting time dependent} decoherence~\cite{repeater2018,
repeater_waitingtime20}. In fact, the longer a qubit waits in a buffer, the
more it decoheres.

To characterise the impact of buffering on quantum communication, researchers
have recently begun to combine queuing models with quantum noise models
\cite{MandayamJagannathanEA20}. In particular, the buffering process inherently
introduces correlations across the noise process experienced by consecutive
qubits, since the waiting times are correlated according to the queuing
dynamics. Thus, to properly characterise the decoherence introduced due to
buffering, we need to look `beyond i.i.d' quantum channels and noise models. 

Although the buffering process leads to correlated noise, it is known to have a
\emph{conditional independence} structure, given the sequence of waiting times
of the qubits. This conditional independence structure can be exploited for
\emph{additive} channels to compute capacity for the correlated noise model, if
the corresponding i.i.d. noise model is well understood in terms of capacity;
see \cite{MandayamJagannathanEA20}.

The generalized amplitude damping channel (GADC) has emerged as an important
model of noise for quantum communication~\cite{shapiro2009, repeater2018}. Even
for the i.i.d case of the  well-studied GADC (see \cite{KhatriSharmaEA20}, for
example), several fundamental questions remain unsolved. For instance, (a) can
the classical capacity of the channel be achieved without entanglement, and (b)
if so, can one construct an explicit encoding-decoding scheme that achieves
capacity? 
 
These questions are well-motivated regardless of any buffering considerations.
Indeed, it is well known that entanglement can be exploited at the encoder and
the decoder for achieving the classical capacity of a quantum channel. For the
class of \emph{additive} channels, the classical capacity can be achieved
without using entanglement at the encoder, although the decoding could involve
joint measurements at the receiver. Performing such joint measurements
typically requires a quantum processor that can carry out quantum gate
operations in a high dimensional Hilbert space.  Since the availability of such
a reliable quantum processor at a communication receiver may not be realistic
in the near future, it is practically relevant to ask after the best achievable
rate without the use of entangled encoding and joint measurements, as well as
the corresponding encoding-decoding.
 
Thus motivated, we make the following contributions in this paper.

\subsection{Our Contributions:}

The GADC is typically parametrized by two quantities, $n$ and $p$. Recent work
has characterized the classical capacity of this channel, for certain parameter
ranges \cite{KhatriSharmaEA20}. In particular, the Holevo information for this
channel has been characterized, which is equal to its classical capacity for
certain parameter values where the channel is known to be additive. 

In the present paper, we first consider a symmetric i.i.d. GADC with $n=1/2,$
and derive the classical capacity of this channel. We do this through the
explicit construction of a symbol-by-symbol encoding at the transmitter and
qubit-by-qubit POVM at the decoder. Specifically, we show that the classical
capacity of a symmetric GADC is achieved without entanglement. For asymmetric
GADC, we characterize the loss in capacity due to non-entangled decoding. To
the best of our knowledge, such  results for GADC have so far been unknown. 

\begin{figure}
    \centering
    \includegraphics[width=16 cm, height = 3.8 cm]{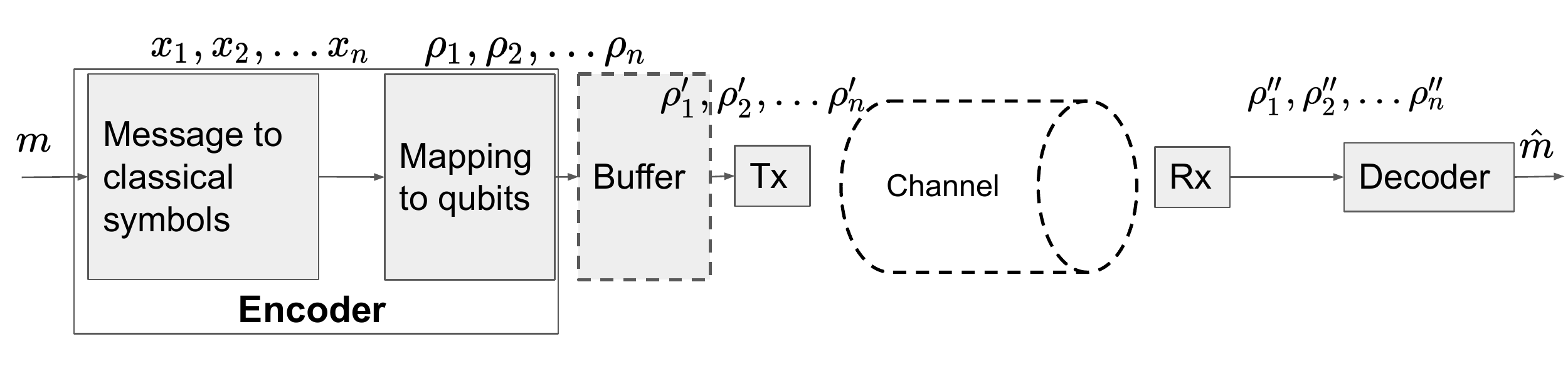}
    \caption{Qubit $\rho_i$ decoheres to $\rho_i'$ while waiting in the buffer
    for transmission. This further decoheres to $\rho''_i$ while passing
    through the channel. Decoherence in the buffer depends on the waiting time
    and results in non-i.i.d. "effective" decoherence.}
    \label{fig:schematic}
\end{figure}

Next, we consider the setting in Fig.~\ref{fig:schematic}, where qubits are
transmitted sequentially, and the qubits decohere as they wait to be
transmitted. The extent of noise suffered by a particular qubit is a function
of the waiting time spent by that qubit. In such a setting, the  "effective"
channel experienced by the qubits is non-stationary and has memory. We model
this waiting time dependent noise using the \emph{quantum queue-channel}
framework, studied in \cite{MandayamJagannathanEA20}. Specifically, we study a
symmetric GAD queue-channel with $n=1/2,$ and the parameter $p$ is made an
explicit function of the waiting time $w$ of each qubit. Such a symmetric GAD
queue channel is known to be additive, which enables the use of the capacity
upper bound obtained in \cite{MandayamJagannathanEA20} for additive queue
channels. Further, we propose a specific encoding for the GAD queue channel,
which induces a binary symmetric classical queue channel. We show that an
achievable rate of this binary symmetric queue channel matches the upper bound
enforced by additivity arguments, thus settling the capacity of the GAD queue
channel, and giving us a fully classical capacity achieving scheme for the
encoder and decoder. Finally, we obtain useful insights for designing practical
quantum communication systems by employing queuing theoretic analysis on the
queue-channel capacity results.

The paper is organized as follows. In the Sec.~\ref{SrelWork} we discuss
related work. To keep this discussion somewhat self-contained, in
Sec.~\ref{SPrelim}, we provide an extended discussion of induced channels,
classical capacities of quantum channels, and non-i.i.d queue-channel
capacities.  In Sec.~\ref{SAmpDamp}, we analyze the generalized amplitude
damping channel~(GADC). Here we discuss the capacity of various induced
channels of GADC~(see Figs.~\ref{figsInd} and~\ref{Fig_ICHolevo}), and prove a
key result~(see Theorem~\ref{thm:CapEqInd}) that Shannon capacity, Holevo
capacity, and the classical capacity of the symmetric~($n=1/2$) GADC are all
equal.  In Sec.~\ref{SqueueCap} we discuss the queue-channel capacity of the
symmetric GADC.  We offer useful design insights by analyzing and numerically
plotting~(see Fig.~\ref{fig:LambdaVsCap}) the capacity expression.
Sec.~\ref{sec:cncl} contains a brief discussion and outlines potentially
interesting future directions.

\subsection{Related Work}
\label{SrelWork}

Our work interleaves different aspects of quantum communication networks, from
quantum Shannon theory to queuing theory.
In quantum Shannon theory, one studies ultimate limits for transmitting
information in the presence of quantum noise.  The generalized amplitude
damping channel~(GADC) is a relevant model of noise in a variety of physical
contexts including communication over optical fibers or free
space~\cite{YuenShapiro78, Shapiro09, ZouLiEA17,RozpedekGoodenoughEA18}, $T_1$
relaxation due to coupling of spins with a high temperature
environment~\cite{ChuangNielsen97, MyattKingEA00,TurchetteMyattEA00}, and
super-conducting based quantum computing~\cite{ChirolliBurkard08}. Quantum
capacities of the i.i.d. GADC have been studied~(see~\cite{KhatriSharmaEA20} and
reference therein). Of particular interest to us are expressions for the Holevo
information of the GADC, found in~\cite{LiZhenMaoFa07} using techniques
from~\cite{Cortese02, Berry05}, and channel parameters~\cite{KhatriSharmaEA20}
where additivity of the GADC Holevo information is known.

While the primary focus of quantum Shannon theory~\cite{Wilde17} has been to
study the classical and quantum capacities of stationary, memoryless quantum
channels~\cite{Holevo12}, recently there has been a spurt of activity in
characterizing the capacities of quantum channels in non-stationary, correlated
settings. We refer to~\cite{RMP2014} for a recent review of the different
capacity results obtained in a context of quantum channels that are not
independent or identical across channel uses. In particular, we focus on the
quantum information-spectrum approach in~\cite{HayashiNagaoka03}, which
provides bounds on the classical capacity of a general, non-i.i.d. sequence of
quantum channels. 

The idea of a quantum queue-channel was originally proposed
in~\cite{qubits_ncc2019} as a way to model and study the effect of decoherence
due to buffering or queuing in quantum information processing tasks. The
classical capacity of quantum queue-channels has been studied for certain
classes of quantum channels, and a general upper bound is known for additive
quantum queue-channels~\cite{MandayamJagannathanEA20}. The effect of
queuing-dependent errors on classical channels has been studied
earlier~\cite{ChatterjeeSeoEA17}, with motivation drawn from crowd-sourcing.
More recently, a dynamic programming based framework for characterising the
queuing delay of quantum data with finite memory size has been proposed
in~\cite{qdelay_2020}. Finally, we note that ideas of queuing theory have also
been used to study aspects of entanglement distribution over quantum networks
such as routing~\cite{EntRouting_2019}, switching, and
buffering~\cite{guha2021}.

\section{Preliminaries}
\label{SPrelim}
\subsection{Classical Channels}
\label{CChan}
\xb
\outl{Classical Channel, i.i.d. noise, transition matrix, channel mutual
information, Shannon entropy, additivity of channel mutual information and
channel capacity} 
\xa

Consider a random variable $X$ that takes discrete values $x$ from a finite set
$\XC$ and another random variable $Y$ that takes discrete values $y$ from some
finite set $\YC$. A {\em discrete memoryless classical channel} $N$ takes an
input symbol $x$ to an output $y$ with conditional probability $p(y|x):=
\Pr(Y=y|X=x)$.  
Sometimes it is convenient to represent a discrete memoryless classical channel
$N$ as a {\em transition matrix} $M$ where $[M]_{yx}$ is $p(y|x)$. 
The channel $N$ is called memoryless because the probability distribution of
the channel output $y$ depends on its current input $x$ and is
conditionally {\em independent} of previous channel inputs. Noise is modelled
by a channel mapping its inputs to a noisy output. When noise on several
different inputs acts in such a way that noise on each input is described by
the same discrete memoryless channel, one says the noise is {\em independent
and identically distributed~(i.i.d.)}. In what follows we discuss rates for
sending information in the presence of i.i.d. noise described by a discrete
memoryless channel $N$.

We define the mutual the {\em channel mutual information},
\begin{equation}
    \IC^{(1)}(N) = \max_{p(x)} I(Y;X),
    \label{chanMutualI}
\end{equation}
where $p(x):=\Pr(X=x)$ is a probability distribution over the input of $N$, and
$I(Y;X)$ is the mutual information given by
\begin{equation}
    I(Y;X) := H(X) + H(Y) - H(X,Y),
\end{equation}
where $H(X) = -\sum_x p(x) \log_2 p(x)$ is the Shannon entropy of the random
variable $X$, and $H(X,Y)$ is the Shannon entropy of $Z=(X,Y)$.  Channel mutual
information~\eqref{chanMutualI} represents an achievable rate for sending
information across $N$. Since $I(Y;X)$ is concave in $p(x)$ for fixed $p(y|x)$,
$\IC^{(1)}(N)$ can be computed numerically with relative
ease~\cite{Arimoto72,Blahut72}. The channel mutual information is additive:
$\IC^{(1)}$ of a channel $N \times N'$ formed by using two channels $N$ and
$N'$ together~(sometimes called ``used in parallel'') is the sum of $\IC^{(1)}$
of the individual channels; that is,
\begin{equation}
    \IC^{(1)}(N \times N') = \IC^{(1)}(N) + \IC^{(1)}(N').
    \label{chanMIAdd}
\end{equation}
Due to additivity, a classical channel $N$'s Shannon capacity,
\begin{equation}
    C_{\text{Shan}}(N) = \lim_{k \mapsto \infty} \frac{1}{k}\IC^{(1)}(N^{\times k}),
    \label{ch1_ChanCap}
\end{equation}
where $N^{\times k}$ represents $k$ parallel uses of the channel $N$,
simplifies to
\begin{equation}
    C_{\text{Shan}}(N) = \IC^{(1)}(N).
    \label{ch1_ChanCap2}
\end{equation}

\subsection{Induced Channels and Classical Quantum Channels}
\xb
\outl{Outline: Induced classical channel, rates for (1) product
encoding-decoding, (2) product encoding and joint decoding. Quantum channel and
rates for each encoding.}
\xa

We now discuss the transmission of classical information using quantum states.
We restrict ourselves to quantum systems described by some finite dimensional
Hilbert space $\HC$. Both pure and mixed states on these quantum systems can be
described using unit-trace positive semi-definite operators~(density operators)
that belong to $\LC(\HC)$, the space of bounded linear operator on $\HC$. 
Suppose classical information is encoded in quantum states using a {\em fixed
map} $\EC : \XC \mapsto \LC(\HC)$, sometimes called a {\em classical quantum}
channel which takes $x \in \XC$ to a density operator $\rho(x) \in \LC(\HC)$.
This classical information can be decoded by a map $\DC : \LC(\HC) \mapsto \YC$
which represents measuring $\rho(x)$ to obtain an output $y \in \YC$.  This
measurement can be described using a POVM~(a collection of positive operators
in $\LC(\HC)$ that sum to the identity). Suppose the POVM $\{\Lm(y)\}$
specifies $\DC$; then any input $x \in \XC$ is decoded as $y$ with conditional
probability,
\begin{equation}
    p(y|x) = \Pr(Y = y | X = x) = \Tr\big( \Lm(y) \rho(x) \big).
    \label{cond1}
\end{equation}
This conditional probability defines an {\em induced channel} $N : \XC \mapsto \YC$.
Induced channels play a vital role in defining the capacity of quantum systems
to send classical information.
For a fixed encoding $\EC$, when using a decoding $\DC$, the induced channel
capacity $C_{\text{Shan}}(N)$ represents a rate at which classical information
can be sent using $\EC$. The maximum rate of this type, sometimes called the
Shannon capacity of $\EC$,
\begin{equation}
    C_{\text{Shan}}(\EC) = 
    \max_{\DC} \IC^{(1)}(\NC) = 
    \max_{\DC, p(x)} I (Y;X),
\end{equation}
is obtained from maximizing the induced channel capacity over all possible
induced channels, i.e., over all possible decoding $\DC$.  Not much is known
about how to perform this optimization. For a fixed $\DC$---that is, fixed
$N$--- $\IC^{(1)}(\NC)$ can be computed with relative ease~(see comments
below~\eqref{chanMutualI}). However, for a fixed $p(x)$ and output alphabet
$\YC$, $I(Y;X)$ is convex in $p(y|x)$, and $p(y|x)$ is linear in the decoding
POVM $\{\Lm(y_j)\}$ specifying $\DC$. The resulting convexity of $I(Y;X)$ in
$\DC$, for fixed $\YC$ makes it non-trivial to numerically compute
$C_{\text{Shan}}(\EC)$.

A $k$-letter message $x^k = (x_1, x_2, \dots, x_k)$ is encoded via $\EC:\XC
\mapsto \LC(\HC)$, sometimes called a {\em product encoding}, into a product
state $\rho_{k}:= \rho(x_1) \ot \rho(x_2) \ot \dots \ot \rho(x_k)$ and decoded
via $\DC : \LC(\HC) \mapsto \YC$, sometimes called a {\em product decoding},
using a product measurement on $\LC(\HC^{\ot k})$. Product decoding is a
special case of joint decoding $\DC_k : \LC(\HC^{\ot k}) \mapsto \YC^{\times
k}$ performed using a joint measurement POVM on $\LC(\HC^{\ot k})$ resulting in
an element on $k$ copies of some classical alphabet $\YC^{\times k}$.  Encoding
$\EC$ followed by joint decoding $\DC_k$ results in an induced channel $N_k$.
Maximizing the channel mutual information $\IC(N_k)$ over all decodings $\DC_k$
defines $\IC^{(k)}(\EC)$.
Due to the presence of entanglement in the joint decoding measurements, one may
have $\IC^{(k)}(\EC) \geq k\IC^{(1)}(\EC)$. This inequality refers to the {\em
super-additivity} of $\IC^{(k)}(\EC)$. Due to super-additivity, a proper
definition of the capacity $C_{pj}$ of sending classical information using {\em
p}roduct encoding $\EC$ and {\em j}oint decoding $\DC_k$ is given by a {\em
multi-letter} formula,
\begin{equation}
    C_{pj} = \lim_{k \mapsto \infty} \frac{1}{n}\IC^{(k)}(\EC).
    \label{HolevoChiOp}
\end{equation}
Remarkably, the Holevo-Schumacher-Westmoreland theorem~\cite{Holevo98,
SchumacherWestmoreland97} gives the above multi-letter expression a {\em
single-letter} form; that is,
\begin{equation}
    C_{pj}(\EC) = \chi^{(1)}(\EC) := \max_{\{p(x)\}} \chi \big( p(x), \rho(x) \big),
    \label{HolevoChi1}
\end{equation}
where the Holevo quantity,
\begin{equation}
    \chi\big( p(x), \rho(x) \big) = S\big(\sum_{x \in \XC} p(x) \rho(x) \big) 
    - \sum_{x} p(x) S\big(\rho(x))\big),
    \label{HolevoChi}
\end{equation}
and $S(\rho) = -\Tr(\rho \log \rho)$, is the von-Neumann entropy of a density
operator $\rho$. Due to the close connection between
$C_{pj}$ and $\chi$, sometimes $C_{pj}(\EC)$ is also denoted by
$C_{\chi}(\EC)$.  There are cases where $C_{\chi}(\EC)$ is strictly greater
than $C_{\text{Shan}}(\EC)$~\cite{SasakiBarnettEA99, Shor04}.  However, much
remains unknown about when and how such separations occurs. 

\subsection{Classical Capacities of a Quantum Channel}
\xb
\outl{Product state capacity of a quantum channel~(product encoding and
joint decoding) and Holevo Quantity. Classical capacity of a quantum
channel~(Joint encoding and decoding), non-additivity of Holevo quantity and
Classical channel capacity expression}
\xa

The quantum analog of a discrete memoryless classical channel is a
(noisy)~quantum channel $\BC$. In general, if $\HC_a$ and $\HC_b$ are two
finite dimensional Hilbert spaces, then the quantum channel $\BC : \LC(\HC_a)
\mapsto \LC(\HC_b)$ is a completely positive trace preserving~(CPTP) map. In
what follows, we discuss transmission of classical information using quantum
states affected by a quantum channel $\BC$~\cite{BennettShor98}~(also see Ch.8
in~\cite{Holevo12}). 

Sending classical information across $\BC$ using product encoding $\EC: \XC
\mapsto \LC(\HC_a)$ and product decoding $\DC: \LC(\HC_b) \mapsto \YC$
results in an induced channel $N: \XC \mapsto \YC$. Maximizing the channel mutual
information of $N$ over $\EC$ and $\DC$ gives the product encoding-decoding
capacity $C_{pp}(\BC)$, also known as the {\em Shannon capacity} of $\BC$,
\begin{equation}
    C_{Shan}(\BC) = \max_{\EC, \DC} \IC(N) = \max_{ \{\EC, \DC \} } \max_{p(x)} I(Y;X).
    \label{ShannonCap}
\end{equation}
The Shannon capacity is bounded from above by the product encoding joint
decoding capacity $C_{pj}(\BC)$, sometimes called the Holevo capacity or the
product state capacity. Using a procedure similar to the one
above~\eqref{HolevoChiOp}, the capacity $C_{pj}(\BC)$ can be defined with the
aid of induced channels generated from product encoding but joint decoding.
Such a definition results in multi-letter expression of the
type~\eqref{HolevoChi1} with a single-letter characterization,
\begin{equation}
    C_{pj}(\BC)= \chi^{(1)}(\BC) := \max_{\{\rho_a(x),p(x)\}} \chi \big(p(x), \rho_b(x) \big),
    \label{HolevoBChi}
\end{equation}
where $\rho_b(x) = \BC(\rho_a(x))$.

The product state capacity $\chi^{(1)}(\BC)$ can be further generalized to
include the possibility of using joint encoding $\EC_k:\XC^{\times k} \mapsto
\LC(\HC_a^{\ot k})$ at the channel input and joint decoding $\DC_k:
\LC(\HC^{\ot k}_b) \mapsto \YC^{\times k}$ at the output. This
encoding-decoding results in an induced channel $\tilde N_k: X^{\times k}
\mapsto Y^{\times k}$.  Maximizing the mutual information of this induced
channel over all encoding $\EC_k$ and decoding $\DC_k$ gives $\tilde
\IC^{k}(\BC)$. Due to the presence of entanglement at the encoding, one may
have super-additivity of the form $\tilde \IC^{k}(\BC) \geq k\tilde \IC(\BC)$.
Due to this super-additivity, the joint encoding-decoding capacity $C_{jj}(\BC)$,
sometimes called the {\em classical capacity} of $\BC$, is defined by a
multi-letter expression of the form~\eqref{HolevoChiOp}.  The capacity
$C_{jj}(\BC)$ can be characterized using the product state
capacity~\eqref{HolevoBChi} as follows,
\begin{equation}
   C_{jj}(\BC) = \chi(\BC) := \underset{k \mapsto \infty}{\lim} \frac{1}{k} \chi^{(1)}(\BC^{\ot k}).
   \label{ch1_HolCap}  
\end{equation}
In general, the limit in~\eqref{ch1_HolCap} is required because the product
state capacity can be non-additive~\cite{Hastings09}; that is, for any two
quantum channels $\BC$ and $\BC'$, the inequality,
\begin{equation}
   \chi^{(1)}(\BC \ot \BC') \geq  \chi^{(1)}(\BC) + \chi^{(1)}(\BC'), 
   \label{ch1_HolNon}  
\end{equation}
can be strict. For certain special classes of channels, the Holevo information
is known to be additive; that is, the inequality above becomes an equality when
$\BC$ is any channel and $\BC'$ belongs to a special class of channels. These
special classes are unital qubit channels~\cite{King02}, depolarizing
channels~\cite{King03}, Hadamard channels~\cite{KingMatsumotoEA07}, and
entanglement breaking channels~\cite{Shor02}.

\subsection{Classical capacity of non-i.i.d. quantum channels}
As mentioned in Sec.~\ref{sec:intro}, the effective channel seen by qubits in
the presence of decoherence  in the transmission buffer is non-i.i.d.
Characterizing the capacity is a harder problem in such a setting. In the
classical setting, a capacity formula for this general non-i.i.d. setting was
obtained using the information-spectrum method~\cite{Han03, VerduHan94}. This
technique was adapted to the quantum setting in~\cite{HayashiNagaoka03}, and a
general capacity formula was obtained for the classical capacity of a quantum
channel. 

\subsubsection{The Quantum inf-information rate}
Recall that a quantum channel is defined as a completely positive,
trace-preserving map $\BC : \HC_a \mapsto \HC_b$ from the "input" Hilbert space
$\HC_a$ to the "output" Hilbert space $\HC_b$. Consider a sequence of quantum
channels $\vec{\mathcal{N}} \equiv \{\mathcal{N}^{(n)}\}_{n=1}^{\infty}$. Let
$\vec{P}$ denote the totality of sequences $\{P^n( X^n ) \}_{n=1}^{\infty}$ of
probability distributions (with finite support) over input sequences $X^n$, and
$\vec{\rho}$ denote the sequences of states $\rho_{X^n}$ corresponding to the
encoding $X^n \rightarrow \rho_{X^n}$. For any $a \in \mathbb{R}^{+}$ and $n$,
we define the operator, \[ \Gamma_{\{P^{n}(X^{n}), \rho_{X^{n}}\}}(a) =
\mathcal{N}^{(n)} (\rho_{X^{n}})- e^{an}\sum_{X^{n}\in
\mathcal{X}^{(n)}}P^{n}(X^{n})\mathcal{N}^{(n)}(\rho_{X^{n}}) . \] Further, let
$\{  \Gamma_{\{P^{n}(X^{n}), \rho_{X^{n}}\}}(a)  > 0\}$ denote the projector
onto the positive eigenspace of the operator $\Gamma_{\{P^{n}(X^{n}),
\rho_{X^{n}}\}}(a) $. 

\begin{definition}
The quantum inf-information rate~\cite{HayashiNagaoka03}
    $\underline{\mathbf{I}}( \, \{ \vec{P}, \vec{\rho} \, \}, \vec{\mathcal{N}}
    \, )$  is defined as,
\begin{equation}
	\underline{\mathbf{I}}( \, \{ \vec{P}, \vec{\rho} \, \}, \vec{\mathcal{N}} \, ) = 
	\sup \left\lbrace a\in\mathbb R^+\left\vert \lim_{n\rightarrow \infty} \sum_{X^{n}\in \mathcal{X}^{(n)}} P^{n}(X^{n}) \tr \left[ \mathcal{N}^{(n)} (\rho_{X^{n}}) \left\lbrace \Gamma_{\{P^{n}(X^{n}), \rho_{X^{n}}\}}(a)   > 0 \right.\right\rbrace \right] = 1 \right\rbrace .\label{eq:quantum_I}
	\end{equation}
\end{definition}

This is the quantum analogue of the classical inf-information rate originally
defined in~\cite{Han03, VerduHan94}. The central result
of~\cite{HayashiNagaoka03} is to show that the classical capacity of the
channel sequence $\vec{\mathcal{N}}$ is given by \[C = \sup_{\{\vec{P},
\vec{\rho}\}}\underline{\mathbf{I}} (\{\vec{P},\vec{\rho}\},\vec{\mathcal{N}}).
\]

\section{Generalized Qubit Amplitude Damping}
\label{SAmpDamp}

The generalized qubit amplitude damping channel~(GADC) $\AC_{p,n} : \LC(\HC_a)
\mapsto \LC(\HC_b)$ is a two parameter family of channels where the parameters
$p$ and $n$ are between zero and one. The channel has a qubit input and qubit
output---$d_a = d_b = 2$--- and its superoperator has the form,
\begin{equation}
    \AC_{p,n}(\rho) = \sum_{i=0}^3 K_i \rho K_i^{\dag},
    \label{ampDampGlm}
\end{equation}
where
\begin{align}
    K_0 &=   \sqrt{1-n}(\dyad{0}{0} + \sqrt{1-p}\dyad{1}{1}), 
    &K_1 &=   \sqrt{p(1-n)}\dyad{0}{1}, \\
    K_2 &=   \sqrt{n}(\sqrt{1-p}\dyad{0}{0} + \dyad{1}{1}), \quad \text{and}
    &K_3 &=   \sqrt{pn}\dyad{1}{0}
    \label{ampDampGlmKr}
\end{align}
are Kraus operators. The GADC~\eqref{ampDampGlm} can also be expressed as
\begin{equation}
    \AC_{p,n} = (1-n)\AC_{p,0} + n \AC_{p,1}.
    \label{ampDampGlmcv}
\end{equation}
The above representation provides the following insightful interpretation.  The
parameter $n$ represents the mixing of $\AC_{p,0}$ with $\AC_{p,1}$, where each
channel $\AC_{p,i}$~($i = 0$ or $1$) is an amplitude damping channel that
favours the state $[i]$ by keeping it fixed and maps the orthogonal state
$[1-i]$ to $[i]$ with damping probability $p$.
When $n = 1/2$, we get equal mixing of both $\AC_{p,0}$ and $\AC_{p,1}$. This
equal mixing represents noise where each state $[i]$~($i = 0,1$) is mapped to
itself with probability $p$ and to $[1-i]$ with probability $1-p$; in other words, this
$n=1/2$ noise treats both $[0]$ and $[1]$ identically. However, when $n$ is not
half, the action of $\AC_{p,n}$ on $[0]$ is different from its action on $[1]$.
In particular, $[0]$ is mapped to itself with probability $1-pn$ and to $[1]$
with probability $pn$, and $[1]$ is mapped to itself with probability $1-p(1-n)$
and to $[0]$ with probability $p(1-n)$.

Any qubit density operator can be written in the Bloch parametrization,
\begin{equation}
    \rho(\rB) = \frac{1}{2} (I + \rB. \vec{\sg}) 
    := \frac{1}{2}(I + x \sg^x + y \sg^y + z \sg^z),
    \label{BlochRep}
\end{equation}
where the Bloch vector $\rB = (x,y,z)$ has norm $|\rB| \leq 1$,
\begin{equation}
    \sg^x = \mat{0 & 1 \\ 1 & 0}, \quad 
    \sg^y = \mat{0 & -i \\ i & 0}, \quad \text{and} \quad
    \sg^z = \mat{1 & 0 \\ 0 & -1}
\end{equation}
are the Pauli matrices, written in the standard basis $\{\ket{0}, \ket{1}\}$.
Using the Bloch parametrization, the entropy 
\begin{equation}
    S\big( \rho(\rB) \big) = h\big( (1 - |\rB|)/2 \big),
    \label{qBitEntro}
\end{equation}
where $h(x):= -[x \log x + (1-x) \log (1-x)]$ is the binary entropy function
and $|\rB| = \sqrt{\rB.\rB}$ is the norm of $\rB$.
An input density operator $\rho(\rB)$ is mapped by $\AC_{p,n}$ to an output
density operator with Bloch vector,
\begin{equation}
    \rB' = (\sqrt{1-p}x, \sqrt{1-p}y, (1-p)z + p(1-2n) ).
    \label{ampGlnBlochVec}
\end{equation}
The GADC is unital at $n=1/2$; that is, $\AC_{p,\frac{1}{2}}(I) = I$. The GADC is
entanglement breaking~\cite{KhatriSharmaEA20} when
\begin{equation}
    2(\sqrt{2} - 1) \leq p \leq 1  \quad \text{and} \quad
    \frac{1}{2}(1 - l(p)) \leq n \leq \frac{1}{2}(1 + l(p)),
    \label{GADCebt}
\end{equation}
where $l(p) = \sqrt{\frac{p^2 + 4p -4}{p^2}}$. The Holevo capacity of unital
qubit channels and entanglement breaking channels is additive; as a result, when
$n=1/2$ and when the values of parameters $p$ and $n$ satisfy~\eqref{GADCebt},
the Holevo information of the generalized amplitude damping channel equals the
classical capacity of the channel. For other values of $p$ and $n$, the
classical capacity of the GADC is not known because for these parameter values,
the Holevo information of the channel is not known to be additive or
non-additive. The actual value of the Holevo information can be computed
numerically.  Next, we briefly discuss this numerical calculation.

\subsection{Holevo Information}
\label{SGADCHolevo}
Let $[\al_+]$ and $[\al_-]$ be projectors on states with Bloch vector
\begin{equation}
    \rB_+ = (\sqrt{1-z^2},0,z), \quad \text{and} \quad 
    \rB_- = (-\sqrt{1-z^2},0,z),
    \label{ensGAPD}
\end{equation}
respectively; here $-1 \leq z \leq 1$. Notice, $[\al_+]$ and $[\al_-]$ are not
orthogonal, except when $z = 0$. It has been shown~\cite{LiZhenMaoFa07} that the
Holevo information,
\begin{equation}
    \chi^{(1)}(\AC_{p,n}) = \max_{\{-1 \leq z \leq 1\}} \; 
    S\big(\AC_{p,n}(\sg)\big) - [S\big(\AC_{p,n}([\al_+])\big) +
    S\big(\AC_{p,n}([\al_-])\big)]/2 ,
    \label{ch5_HolChi}  
\end{equation}
where $\sg = ([\al_+] + [\al_-])/2$. In the above equation, the optimizing
$z$ has the value
\begin{equation}
    z^* =\frac {u - p(1-2n)}{1-p},
    \label{ch5_optZGADC}
\end{equation}
where $u$ comes from solving,
\begin{equation}
    \big(pu - p^2 (1-2n) - p(1-p)(1-2n) \big) f'(r^*) = -r^*(1-\gm) f'(u),
\end{equation}
with
\begin{align}
    f(x)  &:= (1 + x) \log_2 (1+x) + (1-x) \log_2 (1-x), \\
    f'(x) &=  \log_2 \Big( \frac{1+x}{1-x}\Big), \quad \text{and} \\ 
    r^*   &:= \sqrt{1 - p - \frac{\big(u - p(1-2n) \big)^2}{1-p} + u^2 }.
\end{align}
Using the value of $z^*$ in \eqref{ch5_optZGADC} gives,
\begin{equation}
    \chi^{(1)}(\AC_{p,n}) = \frac{1}{2} \big( f(r^*) - \log_2 (1-u^2) -uf'(u) \big).
    \label{GADCHolevo}
\end{equation}
Solving~\eqref{GADCebt} for $n \leq 1/2$ gives a range,
\begin{equation}
    p^* \leq p \leq 1,
    \label{EBTRange}
\end{equation}
where the GADC in entanglement breaking. Here the value,
\begin{equation}
    p^* = \max\Big( 2(\sqrt{2} - 1),  \frac{\sqrt{1 + 4n(1-n)}-1}{2n(1-n)} \Big).
\end{equation}
As indicated earlier, entanglement breaking channels have additive Holevo
capacity.  Thus, when $p$ satisfies \eqref{EBTRange}, the GADC has additive
Holevo capacity. While the Holevo information $\chi^{(1)}(\AC_{n,p})$ gives the
product state classical channel capacity, it doesn't give an explicit encoding
and decoding that achieves this capacity. In what follows, we construct
explicit encoding and decodings---in other words, we construct induced classical channels,
and compare the capacity of these channels to the product state classical
capacity $\chi^{(1)}(\AC_{n,p})$. For $n=1/2$, we find the optimal encoding
and decoding which achieves $\chi^{(1)}(\AC_{1/2,p})$ for all $0 \leq p \leq 1$.

\subsection{Induced Channels}
\label{SsubICAmpGln}
To obtain an induced channel for $\AC_{p,n}$ one must choose an encoding and
decoding. To choose an encoding, $\EC:x \mapsto \rho(x)$, one fixes a set
of input states $\{\rho(x)\}$.  To choose a decoding, $\DC: \rho(x) \mapsto
y$, one fixes an output measurement POVM $\{\Lm(y)\}$. Together the
encoding-decoding results in an induced channel with conditional probability
$p(y|x) = \Tr(\rho(x) \Lm(y))$.
A priori, there is no clear choice for these input states and output
measurement. However, the generalized qubit amplitude damping channel satisfies
an equation
\begin{equation}
    \AC_{p,n}\big(\sg^z_a \; \rho \; (\sg^z_a)^{\dag} \big)  = \sg^z_b \; \AC_{p,n}(\rho) \; (\sg^z_b)^{\dag},
    \label{ampCov}
\end{equation}
where the subscripts $a$ and $b$ on the Pauli operator $\sg^z$ signify the
space on which the operator acts. The above equation implies that the
generalized amplitude damping channel has a rotational symmetry around the
$z$-axis.
Using this rotational symmetry and the fact that $\AC_{p,n}$ is a qubit
input-output channel one may choose an encoding $\EC:x \mapsto \rho(x)$
where $x = 0$ or $1$ and $\{\rho(x)\}$ are two orthogonal input states that
remain unchanged under the $\sg^z_a$ symmetry operations; that is,
$\rho(x) = [x]$.  To decode, one may apply a protocol
for correctly identifying a state chosen uniformly from a set of two known
states $\AC_{p,n}([0])$ and $\AC_{p,n}([1])$ with highest probability.
This protocol comes from the theory of quantum state
discrimination~\cite{Helstrom69}. It uses a POVM with two elements $\{E, I_b -
E\}$, where $E$ is a projector onto the space of positive eigenvalues of
$\AC_{p,n}([0]) - \AC_{p,n}([1])$. An unknown state, either $\AC_{p,n}([0])$ or
$\AC_{p,n}([1])$ with equal probability, is measured using the POVM. If the
outcome corresponding to $E$ occurs, the unknown state is guessed to be
$\AC_{p,n}([0])$; otherwise, the guess is $\AC_{p,n}([1])$. In the present
case, a simple calculation shows that $E = [0]$.

\begin{figure}[ht]
    \begin{subfigure}{.5\textwidth}
        \includegraphics[scale=.5]{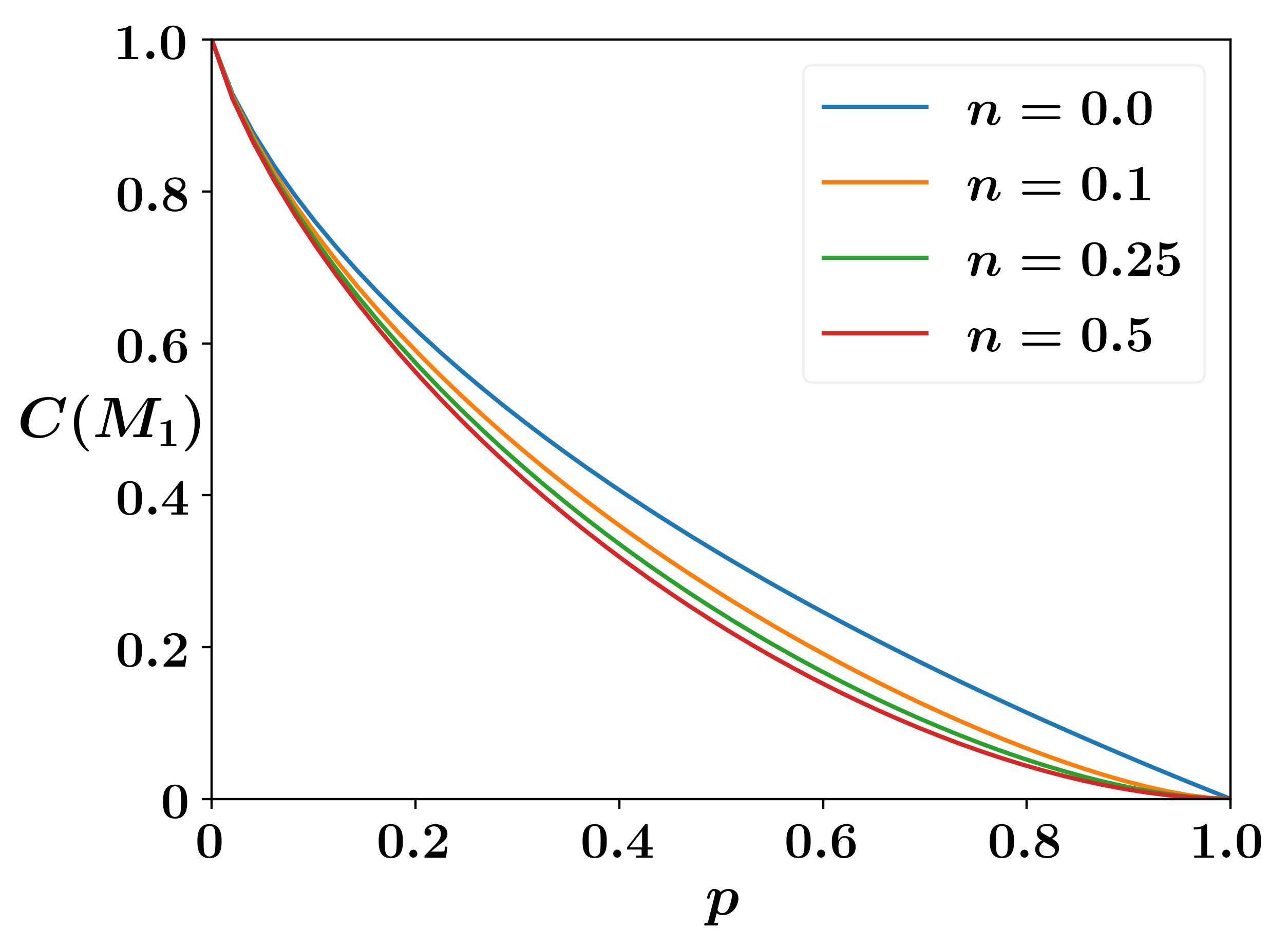}
        \caption{The capacity $C(M_1)$ of the induced channels $M_1$ \newline as a
        function of $p$ for various values of the parameter $n$. \newline 
        This channel
        $M_1$ is defined by transition probability \newline matrix $P'$ in
        \eqref{IC1AmpGln}.\label{ICAmp2Caps}}
    \end{subfigure}
    \begin{subfigure}{.5\textwidth}
        \includegraphics[scale=.5]{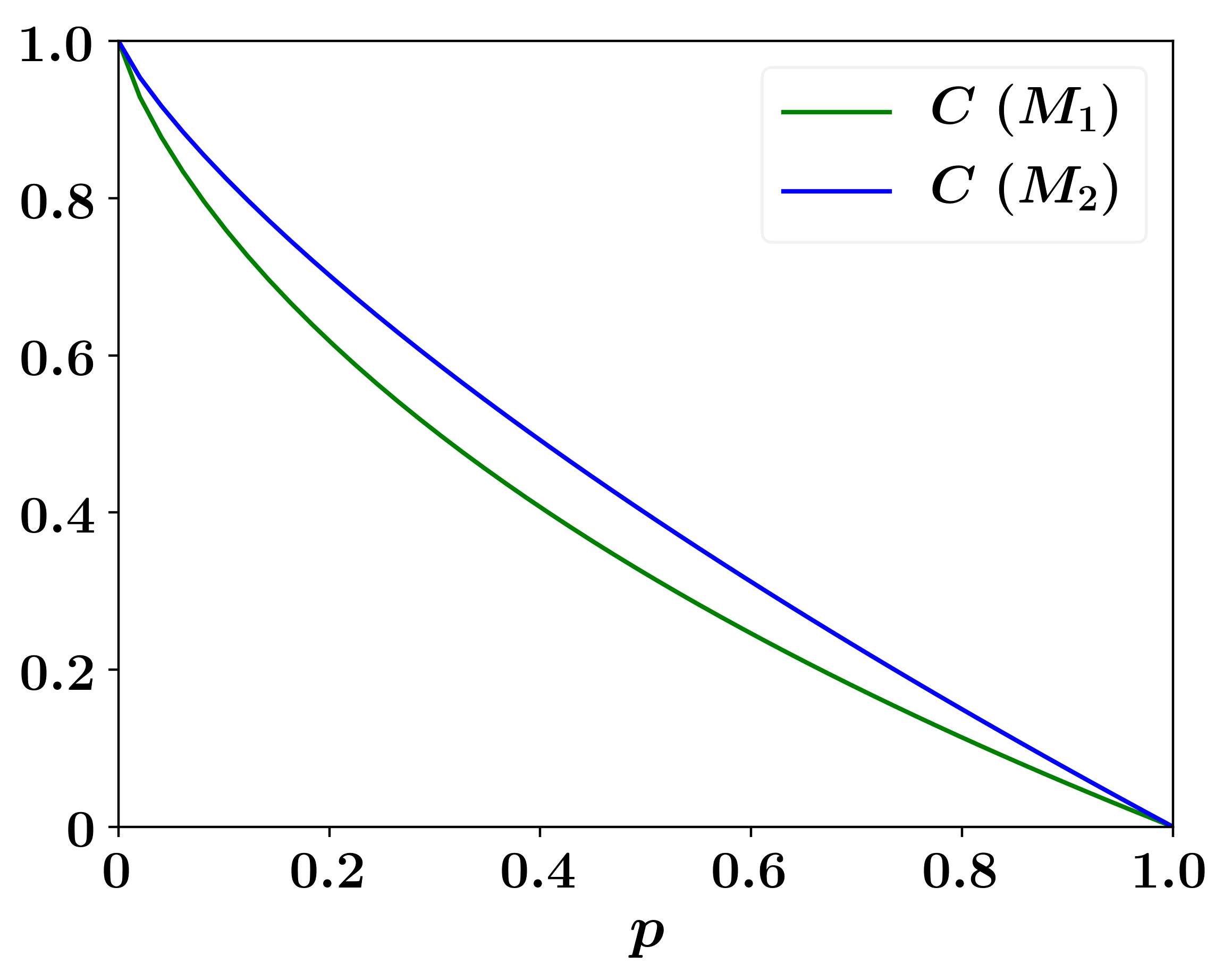}
        \caption{As a function of $p$, the capacity $C(M_1)$ is largest when $n
        = 1/2$. This largest capacity is plotted in green and $C(M_2)$ is
        plotted in blue. The induced channels $M_1$ and $M_2$, are defined by
        transition probability matrices $P'$ and $Q(0)$ in \eqref{IC1AmpGln}
        and \eqref{IC2Amp}, respectively.  \label{ICAmp1Caps}}
    \end{subfigure}
    \caption{Capacities of induced channels~\label{figsInd}}
\end{figure}

The state discrimination protocol outlined above can be used as the decoding
map $\DC: \rho(x) \mapsto y$, which measures $\rho(x)$ using the POVM $\{\Lm(y)
= [y]\}$ and returns $y \in \{0,1\}$ with conditional probability $\Tr(\rho(x)
\Lm(y))$. This choice of decoding, together with the encoding, $\EC:x \mapsto
[x]$, results in an induced channel $M_1$ with transition probability matrix
\begin{equation}
    P' = \mat{1-pn & p(1-n) \\
             pn & 1 - p(1-n)}.
\label{IC1AmpGln}
\end{equation}
This induced channel $M_1$ corresponding to the above matrix $P'$ is a binary
{\em asymmetric} channel that flips $x=00$ to $x=1$ with probability $pn$ but
flips $x=1$ to $x=0$ with probability $p(1-n)$. The capacity of $M_1$ is given by
\begin{equation}
    C(M_1) = \max_{0 \leq a \leq 1 } h(r) - (1-a) h( pn ) - ah\big(p(1-n)\big),
    \label{IC1Amp2Cap}
\end{equation}
where $r = (1-pn) - a(1-p)$.  From the above equation, it is clear that $C(M_1)$
is unchanged when $n$ is replaced with $1-n$. We restrict ourselves to $0 \leq
n \leq 1/2$. As can be seen from Fig.~\ref{ICAmp2Caps}, $C(M_1)$ decreases with
$n$ for $0 \leq n \leq 1/2$.

Next, we consider a different induced channel where the encoding is performed
using possibly non-orthogonal states and decoding is performed using a measurement
designed to distinguish these encoded states at the channel output with
maximum probability.  The encoding maps $x=0$ and $x=1$ to $[\al_+]$ and
$[\al_-]$~(defined via eq.~\eqref{ensGAPD}), respectively. The decoding is
performed using a POVM $\{ \Lm(y) \}$ where $\Lm(y)$ at $y=0$ is the
projector onto the space of positive eigenvalues of $\AC_{p,n}([\al_+]) -
\AC_{p,n}([\al_-])$. This projector is simply $[x+]$, where $\ket{x+} :=
(\ket{0} + \ket{1})\sqrt{2}$. This encoding-decoding scheme results in
a one-parameter family of induced channels $M_2(z)$. This family of channels
has a transition matrix
\begin{equation}
    Q(z) = \mat{q(z) & 1-q(z) \\
             1-q(z) & q(z)},
\label{IC2Amp}
\end{equation}
where $q(z) = (1- a(z)\sqrt{1-p})/2$, and $a(z) = \sqrt{1-z^2}$. For any $z$,
the induced channel $M_2(z)$ is a binary symmetric channel with flip
probability $q(z)$. Interestingly, this family of induced channels coming from
the GADC $\AC_{p,n}$ does not depend on the parameter $n$ in the channel. The
Shannon capacity of $M_2(z)$ is simply
\begin{equation}
    C(M_2(z)) = \; 1 - h \big(q (z) \big).
    \label{M2Shannon}
\end{equation}
For a fixed $p$, one can easily show that $C(M_2(z))$ is maximum when $z = 0$;
thus, $M_2 =M_2(0)$ has the largest Shannon capacity among the one-parameter
family of induced channels $M_2(z)$. This induced channel $M_2$ is simply a
binary symmetric channel with flip probability $q = (1 - \sqrt{1-p})/2$.  We
compare the capacity of $M_2$ with that of $M_1$ at $n=0$, defined earlier~(see
Fig.~\ref{ICAmp1Caps}) to find that 
\begin{equation}
    C(M_2) \geq C(M_1),
    \label{indComp}
\end{equation}
for all $p$ and $n$. In Fig.~\ref{Fig_ICHolevo}, we compare the capacity
$C(M_2)$ of the induced channel $M_2$ with the Holevo information of the GADC
for various values $n$.  We numerically find that for values of $n < 1/2$ and
$0 < p < 1$, $C(M_2)< \chi^{(1)}(\AC_{p,n})$. Next we focus on $\AC_{p,1/2}$,
the GADC at $n=1/2$.

\begin{theorem}
\label{thm:CapEqInd}
For the $n=1/2$ GADC, $\AC_{p,1/2}$, the Shannon capacity, product-state
capacity, and classical capacity are all equal to that of a binary
symmetric channel with flip probability $q = (1 - \sqrt{1-p})/2$.
\end{theorem}

\begin{proof}
    Our proof has two key ingredients: first is a well-known additivity of
    $\chi^{(1)}(\AC_{p,1/2})$ and second is an argument to show that capacity
    of a binary symmetric induced channel of $\AC_{p,1/2}$ bounds
    $\chi^{(1)}(\AC_{p,1/2})$ from above and below. 
    
    Notice the $n=1/2$ GADC, $\AC_{p,1/2}$ is unital. As a result, the channel's
    Holevo information is additive and equals the channel's classical capacity;
    that is,
    \begin{equation}
        \chi(\AC_{p,1/2}) = \chi^{(1)}(\AC_{p,1/2}).
        \label{addEq1}
    \end{equation}
    The Holevo information $\chi^{(1)}(\AC_{p,1/2})$ is the product-state
    classical capacity of $\AC_{p, 1/2}$. This product state capacity bounds
    the Shannon capacity of $\AC_{p,1/2}$ from above. In turn, this Shannon
    capacity upper bounds the capacity of any induced channel of $\AC_{p,1/2}$
    that uses product encoding and decoding. Since the induced channel $M_2$,
    defined in~\eqref{M2Shannon}, uses product encoding and product decoding,
    \begin{equation}
        C(M_2) \leq C_{\text{Shan}} (\AC_{p,1/2}) \leq \chi^{(1)}(\AC_{p,1/2}).
        \label{lowerCM2}
    \end{equation}
This induced channel $M_2$ is a binary symmetric channel with flip probability
$q = (1 - \sqrt{1-p})/2$. The channel's capacity is
\begin{equation}
    C(M_2) = 1 - h(q),
    \label{capM2}
\end{equation}
where $h$ is the binary entropy function.  We now show this capacity $C(M_2)$
bounds $\chi^{(1)}(\AC_{p,1/2})$ from above. From eq.~\eqref{ch5_HolChi},
\begin{equation}
    \chi^{(1)}(\AC_{p,1/2})  = \max_{-1 \leq z \leq 1} \{
    S\big(\AC_{p,1/2}(\sg)\big) - [S\big(\AC_{p,1/2}([\al_+])\big) +
    S\big(\AC_{p,1/2}([\al_-])\big)]/2 \}.
    \label{chiExp1}
\end{equation}
Using~\eqref{ensGAPD},\eqref{ampGlnBlochVec}, and \eqref{qBitEntro} we find
that 
\begin{equation}
    S\big(\AC_{p,1/2}([\al_+])\big) =
    S\big(\AC_{p,1/2}([\al_-])\big) = h\big((1- |\rB_b|)/2\big),
    \label{ents}
\end{equation}
where $\rB_b = \big( \sqrt{(1-p)(1-z^2)}, 0, (1-p)z \big)$. Using~\eqref{ents}
and~\eqref{chiExp1} we bound the value of $\chi^{(1)}(\AC_{p,1/2})$ from
above as follows:
\begin{equation}
    \chi^{(1)}(\AC_{p,1/2})  \leq \max_{-1 \leq z \leq 1}
    S\big(\AC_{p,1/2}(\sg)\big) - \min_{-1 \leq z \leq 1} h\big((1 - |\rB_b|)/2\big)
    \label{chiExp2}
\end{equation}
Notice (1) the maximum value of $S\big(\AC_{p,1/2}(\sg)\big)$ is at most $1$
because $\AC_{p,1/2}$ has a qubit output; (2) $h((1 - |\rB_b|)/2)$ is
monotonically decreasing in $|\rB_b|$ and takes its minimum value $h\big((1
-\sqrt{1-p})/2\big)$ at $z=0$. Using these two facts in~\eqref{chiExp2} along
with~\eqref{capM2},
\begin{equation}
    \chi^{(1)}(\AC_{p,1/2}) \leq 1 - h(q) = C(M_2).
    \label{chiUpper}
\end{equation}
Together,~\eqref{chiUpper}, \eqref{lowerCM2}, and \eqref{addEq1} prove
\begin{equation}
    C(M_2) = C_{Shan}(\AC_{p,1/2}) = \chi^{(1)}(\AC_{p,1/2}) = \chi(\AC_{p,1/2}) = 1 - h(q),
    \label{eqCap}
\end{equation}
where $q = (1 - \sqrt{1-p})/2$.
\end{proof}

This equality~\eqref{eqCap} above shows that the induced channel $M_2$, obtained
from product encoding and product decoding, achieves not only the Shannon
capacity $C_{Shan}(\AC_{p,1/2})$ but also the product state capacity of
$\AC_{p,1/2}$.  A by-product is an alternate proof for the product state
capacity $\chi^{(1)}(\AC_{1/2,p})$ expression~(see Ex.8.1 in~\cite{Holevo12}).
Even more notably, the product state capacity in general allows for joint
decoding of its product state inputs; however, we find that product decoding of
the type in the induced channel $M_2$ suffices to achieve this capacity. For
$\AC_{p,1/2}$ the product state capacity $\chi^{(1)}(\AC_{p,1/2})$ is additive
and equals the ultimate channel capacity $\chi(\AC_{p,1/2})$. This ultimate
capacity allows for the more general joint encoding and joint decoding, yet the
additivity of $\chi^{(1)}(\AC_{p,1/2})$, along with the equality
in~\eqref{eqCap}, show how this ultimate capacity is simply achieved using
product encoding and product decoding.

\begin{figure}[ht]
    \begin{center}
        \includegraphics[scale=.75]{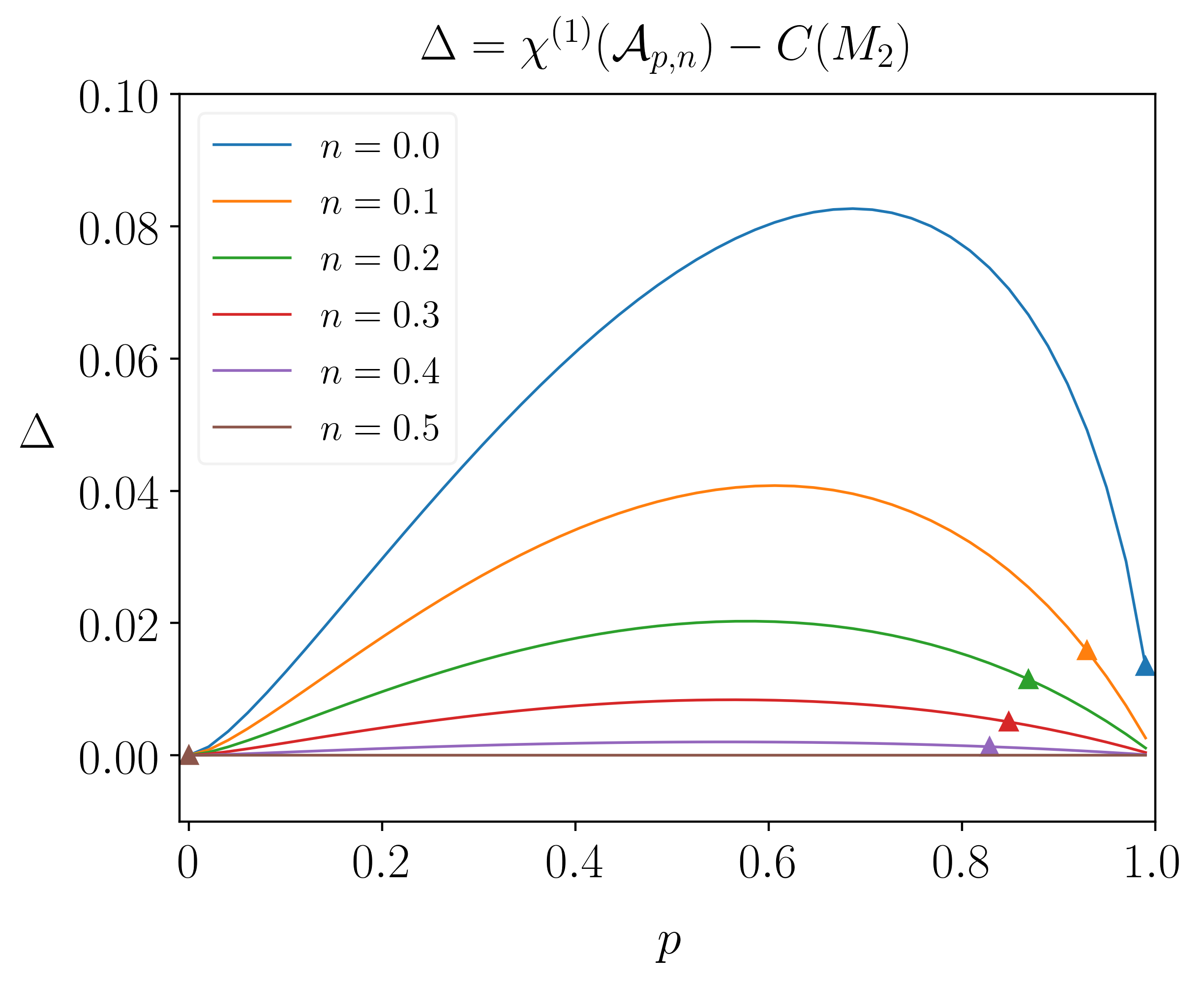}
        \caption{The difference $\Dl = \chi^{(1)}(\AC_{p,n}) - C(N_2)$ as a
        function of $p$ for various values of the parameter $n$. For each $n$,
        the colored triangle indicates the value of $p^*$ above which 
        $\chi^{(1)}(\AC_{p,n})$ is additive.\label{Fig_ICHolevo}}
    \end{center}
\end{figure}

In what follows, we focus on the $n=1/2$ GADC $\AC_{p,1/2}$. As discussed
below~\eqref{ampDampGlmcv}, this channel describes noise in which both
computational basis states $\ket{0}$ and $\ket{1}$ are treated on equal
footing. When information about which of these computational basis states
decays faster than the other, the GADC with $n\neq 1/2$ is an apt noise model.
However when such information is unavailable, or when it is known that both
computational basis states decay but the maximally mixed state doesn't, one
uses the $n=1/2$ GADC. One simple example of such noise is the qubit thermal
channel~(analogous to the bosonic thermal channel~\cite{KhatriSharmaEA20,
MyattKingEA00,TurchetteMyattEA00}) in which the channel environment is
represented by the maximally mixed state. Another simple example is the effect
of dissipation to an environment at a finite temperature~\cite{NielsenChuang11}.

\section{Decoherence in Buffer: A Queue-channel Approach}
\label{SqueueCap}
For  communicating a message over a GAD channel, a sequence of product quantum
states has to be transmitted serially. In an idealized i.i.d. setting, it is
implicitly assumed that the transmission of each state takes a fixed amount of
time and that the qubits are  prepared accordingly to avoid any buffering.
However, in practice,  preparation times as well as transmission (and
reception) times are stochastic due to the inherent quantum physical nature of
the devices. This naturally leads to {\em queuing} of qubits at the buffer of
the transmitter, and the queued qubits decohere due to their interactions with
the environment of the buffer. A schematic of this setting is presented in
Fig.~\ref{fig:schematic}.

Note that the decoherence while waiting in the buffer is in addition to the
decoherence experienced by the qubits while passing through the channel.  The
decoherence experienced by a qubit in the buffer depends on its waiting time in
the buffer: the longer the wait, the worse the decoherence. Due to the queuing
dynamics, the waiting times of the qubits are different and non i.i.d. Hence,
the effective decoherence experienced by the qubits are also non i.i.d..

Most information theoretic channel models assume an idealized scenario without
any decoherence in the  transmission buffer. In Sec.~\ref{SAmpDamp}, we studied
the capacity of GAD channels assuming such an idealized scenario and showed
that non-entangled encoding and decoding schemes can  achieve the capacity
of symmetric ($n=\frac{1}{2}$) GADC. In this section, we introduce a GAD {\em
queue-channel} that can model both channel and buffer decoherence. The concept
of a  {\em quantum queue-channel} was introduced in
\cite{MandayamJagannathanEA20} by adapting a related classical notion from
\cite{ChatterjeeSeoEA17}. Building on  the results in Sec.~\ref{SAmpDamp}, we
characterize the capacity of  GAD queue-channels and  show that non-entangled
encoding and decoding schemes achieve that capacity.

\subsection{Symmetric GAD queue-channel}

For a symmetric GADC, the parameter $p$ captures the level of damping
experienced by a qubit while interacting with an environment. In the absence
of buffer decoherence, $p$ depends on the flight time $T_f$ through the
channel and the physical parameters of the channel. Similarly, the level of
damping experienced in the buffer depends on the waiting time in the buffer
$W$ and the physical parameters of the buffer. Hence, the effective GADC
parameter experienced by a qubit is a function $g(T_f,W)$ of its waiting time
and its flight time, where the form of $g(\cdot)$ depends on the physical
parameters of the channel and the buffer. As the flight time is almost
deterministic, for simplicity of notations we denote this function by
$p_{\mbox{eff}}(W)$. 

In optical quantum communication, the prevalent form of quantum communication,
often the transmission channel and the buffer are both made of optical fibers.
Also, in practice, it is mostly the case that the noise treats both $[0]$ and
$[1]$ identically. So, it is natural to assume that the mixing parameter $n$ is
$\frac{1}{2}$ and $n$ does not depend on the waiting time. On the other hand,
the damping parameter $p$ depends on the time of interaction with the
environment. These motivated the above modeling assumptions regarding symmetric
GAD queue-channel with waiting time dependent damping parameter
$p_{\mbox{eff}}(W)$. However, for other modes of quantum communication, the
models for buffer decoherence and channel decoherence can be different. This
would be an interesting direction for future exploration.
 
A classical message is encoded into a sequence of classical states $x_1, x_2,
\ldots, x_k$, which in turn is transmitted as a sequence of quantum states or
qubits $\rho_1, \rho_2, \ldots, \rho_k$. Effectively, qubit $i$ is received at
the receiver after its passage through a symmetric GADC with parameter
$p_{\mbox{eff}}(W_i)$, where $W_i$ is the time between the preparation of the
$i$th qubit and its transmission. Given the knowledge of $W_1, W_2, \ldots,
W_k$ at the receiver, qubits experience  independent but not identically
distributed symmetric generalized amplitude damping decoherence. 
 
We complete the description of the combined channel with the mathematical model
of the queuing system that gives rise to the waiting time sequence. The
buffering process is modeled as a continuous-time single-server queue. To be
specific, the single-server queue is characterised by (i) A server that
processes the qubits in the order in which they arrive, that is in a First Come
First Served (FCFS) fashion\footnote{The FCFS assumption is not required for
our results to hold, but it helps the exposition.}, and (ii) An "unlimited
buffer" --- that is, there is no limit on the number of qubits that can  wait
to be transmitted. We denote the time between preparation of the $i$th and
$i+1$th qubits by $A_i$, where $A_i$ are i.i.d. random variables. These $A_i$s
are viewed as inter-arrival times of a point process of rate $\lambda,$ where
$\mathbb E[A_i]=1/\lambda.$ The "service time," or the time taken to transmit
qubit $i,$ is denoted by $S_i$, where $\{S_i\}$ are also assumed to be i.i.d.
random variables, independent of the inter-arrival times $A_i,i\geq 1$. The
"service rate" of the qubits is denoted by $\mu=1/\mathbb E[S_i].$ We assume
that $\lambda<\mu$ (i.e., mean transmission time is strictly less than the mean
preparation time) to ensure stability of the queue. Qubit $1$ has a waiting
time $W_1=S_1$. The waiting times of the other qubits can be obtained using the
well known Lindley's recursion: \[W_{i+1} = \max(W_i - A_i, 0) + S_{i+1}.\] In
queuing parlance, the above system describes a continuous-time $G/G/1$ queue.
Under mild conditions, the sequence $\{W_i\}$ for a stable $G/G/1$ queue is
\emph{ergodic,} and reaches a \emph{stationary distribution} $\pi.$ We assume
that the waiting times $\{W_i\}$ of the qubits are  {\em not} available at the
transmitter during encoding, but are available at the receiver during decoding. 

An important difference between the queue-channel introduced above and the
usual i.i.d. channels is that this channel is a part of continuous time
dynamics. Hence, the usual notion of capacity per {\em channel use} for i.i.d.
channels is not pertinent here.  As mentioned before, the above channel model
is closely related to quantum queue-channels studied in
\cite{MandayamJagannathanEA20}. So, we first do a short review of the notion of
capacity per {\em unit time} and some relevant capacity results in
\cite{MandayamJagannathanEA20}.
 
 \subsubsection{Classical capacity of additive quantum queue-channels}
\begin{definition}
\label{def:contAchievableRate}
A rate $R$ is called an achievable rate for a quantum queue-channel if there
exists a sequence of $(n, 2^{R T_n})$ quantum codes with probability of error
$P_e^{(n)} \to 0$ as $n\to \infty$ and $\mathbf{E}\left[\sum_{i=1}^{n-1}
    A_i+W_n\right] \le T_n$.
\end{definition}

\begin{definition}
\label{def:capacity}
The information capacity of the queue-channel is the supremum of all achievable
rates for a given arrival and service process, and is denoted by $C$ bits per
    unit time.
\end{definition}
Note that the information capacity of the queue-channel depends on the arrival
process, the service process, and the noise model. We assume that the receiver
knows the realizations of the arrival and the departure times of each symbol, a
realistic assumption in several physical scenarios, as discussed
in~\cite{MandayamJagannathanEA20}.

\begin{definition}[Additive quantum queue-channel] A quantum queue-channel $\vec{\mathcal{N}}_{\vec{W}}$ is said to be additive if the Holevo information of the underlying single-use quantum channel $\mathcal{N}$ is additive. Specifically, additivity of the Holevo information of the quantum channel $\mathcal{N}$ implies
\[ \chi^{(1)} (\mathcal{N}_{W_{1}} \otimes \mathcal{N}_{W_{2}}) = \chi^{(1)} (\mathcal{N}_{W_{1}}) + \chi^{(1)} (\mathcal{N}_{W_{2}}) . \] 
\end{definition}

\begin{proposition}
\label{prop:expression}
The capacity of the quantum queue-channel (in bits/sec) is given by,
\begin{align}
C = \lambda \sup_{\{\vec{P}, \vec{\rho}\}}\underline{\mathbf{I}} (\, \{\vec{P},\vec{\rho}\},\vec{\mathcal{N}}_{\vec{W}} \, ),
\label{eq:cap_exp2}
\end{align}
where,  $\underline{\mathbf{I}}( \, \{ \vec{P}, \vec{\rho} \, \}, \vec{\mathcal{N}}_{\vec{W}} \, )$ is the quantum spectral inf-information rate defined in Eq. \ref{eq:quantum_I}.
\end{proposition}

We conclude this section by stating the general upper bound for the capacity of additive quantum queue-channels, proved in~\cite{MandayamJagannathanEA20}. 
\begin{theorem}\label{thm:queueCapacity_ub}
For an additive quantum queue-channel $\vec{\mathcal{N}}_{\vec{W}}$, the capacity is bounded as, 
\[ C \leq \lambda~\EX_{\pi}\left[ \chi^{(1)}(\mathcal{N}_{W}) \right]\ {\rm bits/sec.}, \] 
where, $\mathbf{E}_{\pi}$ is expectation with respect to the stationary distribution $\pi$ of $\{W_i\}$. Here $\chi^{(1)}(\mathcal{N}_{W})$ denotes the Holevo information of the single-use  quantum queue-channel corresponding to waiting time $W$.
\end{theorem}

 \subsection{Capacity of the symmetric GAD queue-channel}
 To characterize the capacity of the above channel, we use the additive queue-channel capacity results from \cite{MandayamJagannathanEA20}. First, we present the converse result, that is the capacity upper bound. 
 \begin{corollary}
 \label{cor:qcUB}
 The capacity of a symmetric GAD queue-channel is upper bounded by 
 \[\lambda~\mathbf{E}_{\pi}\left[1-h\left(\frac{1-\sqrt{1-p_{\mbox{\em eff}}(W)}}{2}\right)\right].\]
 \end{corollary}
 This corollary follows directly from \cite{MandayamJagannathanEA20}[Theorem 1] using  Eq. \ref{eqCap} in Sec.~\ref{SsubICAmpGln}
 \[1-h\left(\frac{1-\sqrt{1-p_{\mbox{eff}}(W)}}{2}\right)=\chi^{(1)}\left(\AC_{p_{\mbox{\small eff}}(W),\frac{1}{2}} \right).\] 
 This is because the symmetric GAD queue-channel is an additive queue-channel and the upper bound in  \cite{MandayamJagannathanEA20}[Theorem 1] (stated here as Theorem~\ref{thm:queueCapacity_ub}) is valid for any additive queue-channel.

 To prove achievability, we build on the results presented in Sec.~\ref{SsubICAmpGln} and obtain the following result.
 
 \begin{theorem}\label{thm:qcAchieve}
 There exists a product encoding and a non-entangled decoding scheme that achieves a rate of 
 \[\lambda~\mathbf{E}_{\pi}\left[1-h\left(\frac{1-\sqrt{1-p_{\mbox{\em eff}}(W)}}{2}\right)\right]\]
 over a symmetric GAD queue-channel.
 \end{theorem}
 
 \begin{proof}[Proof (sketch)]
 Proof of Theorem~\ref{thm:qcAchieve} follows using an induced channel approach. For proving achievability using non-entangled encoding and decoding, we use the product encoding and decoding similar to channel $M_2$ in Sec.~\ref{SsubICAmpGln}, excpt that the POVM for the $i$th qubit uses $p_{\mbox{eff}}(W_i)$ instead of $p$ in Eq. \ref{IC2Amp}. As $M_2$ is a binary symmetric channel, the rest of the proof of Theorem~\ref{thm:qcAchieve} follows using steps very similar to the proof of \cite{MandayamJagannathanEA20}[Theorem~4]. 
 \end{proof}

\subsection{Useful design insights} 
As the motivation for this work is the practical issues faced by current quantum networks, we end with some practical insights obtained from the analytical results.

In the idealized i.i.d. setting, the capacity per {\em unit time} for a given $\lambda$ is simply $\lambda$ times the capacity per {\em channel use}. For the stability of the transmission buffer, a necessary and sufficient condition is $\lambda<\mu$. These two facts together seem to imply that in practice, the preparation/transmission rate $\lambda$ should be close to $\mu$ for achieving high data rate (per unit time). However, as we show below, this is an erroneous design choice that may lead to highly sub-optimal rates.

In the presence of buffering and decoherence in the buffer, the optimal qubit preparation/transmission rate can be obtained by optimizing the capacity expression in Theorem~\ref{thm:qcAchieve}: \[\arg\max_{\lambda\in(0,\mu)} \lambda~\mathbf{E}_{\pi}\left[1-h\left(\frac{1-\sqrt{1-p_{\mbox{eff}}(W)}}{2}\right)\right].\]
Though it may appear that the capacity expression increases with $\lambda$, it is not so since $\pi(\cdot)$ depends on $\lambda$. 

In general, obtaining a closed form expression for the best $\lambda$ is not possible. We consider a simple setting where $\{S_i\}$ and $\{A_i\}$ are i.i.d. exponential random variables with mean $1$ and $\lambda^{-1}$ ($>1$), respectively. This is in fact the well known M/M/1 queue setting, where $\pi$ turns out to be exponential distribution with mean $(1-\lambda)^{-1}$. We assume the prevalent exponential decoherence model for decoherence in the buffer, that is for $\kappa>0$
\begin{align}p_{\mbox{eff}}(W) = 1 - \exp(-\kappa~W). \label{eq:expDecoh}\end{align}

\begin{figure}
    \centering
    \includegraphics[scale=0.4]{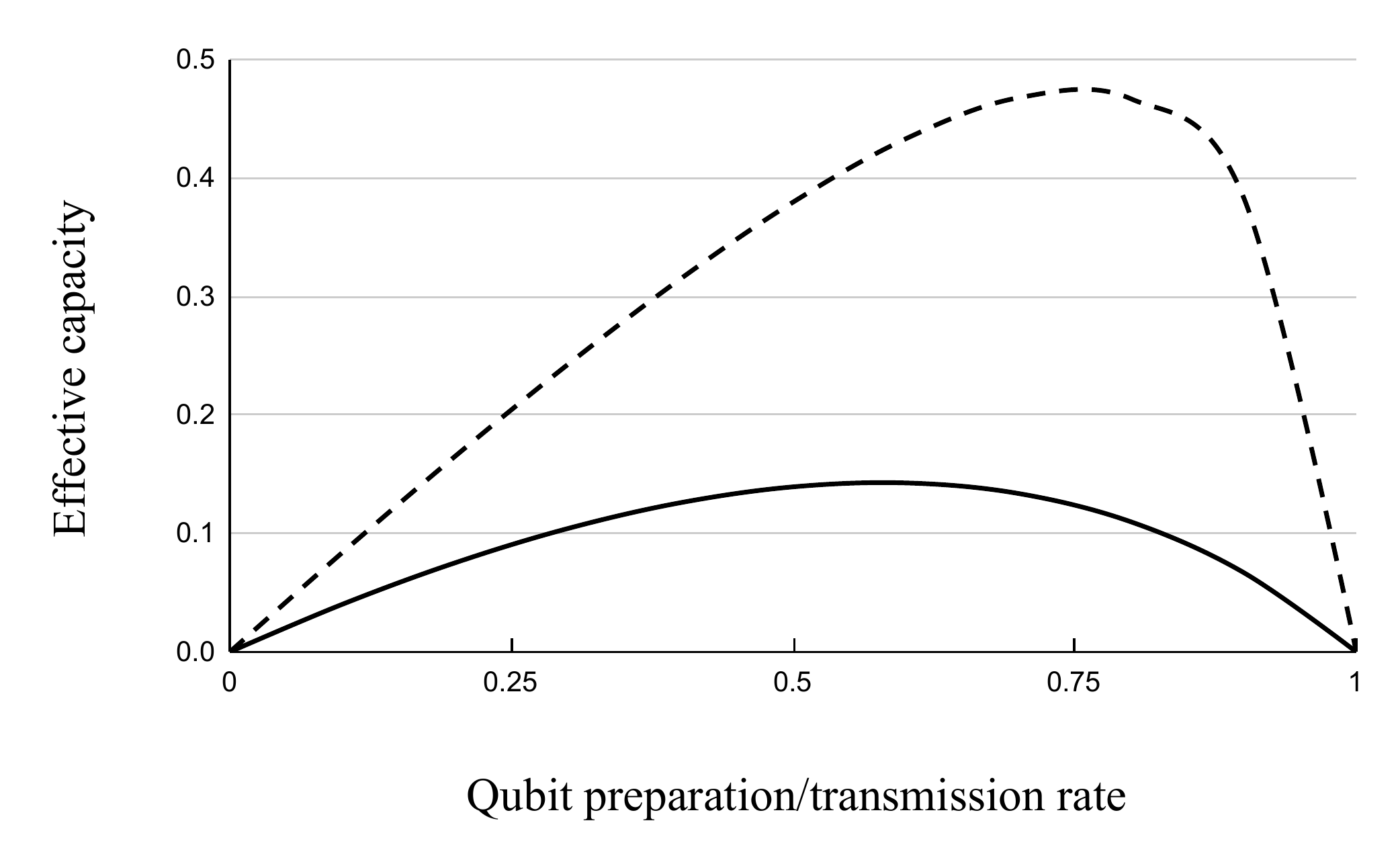}
    \caption{Capacity vs $\lambda$ for $\kappa=1$ (solid) and $\kappa=0.1$ (dotted).}
    \label{fig:LambdaVsCap}
\end{figure}

In Fig. \ref{fig:LambdaVsCap}, the capacity expression in Theorem~\ref{thm:qcAchieve} is plotted against $\lambda$ for $\kappa=1$  and $\kappa=0.1$. Intuitively, the best $\lambda$ for a given $\kappa$ is the point where the curve reaches its peak.

Clearly, the optimal $\lambda$ is not close to $\mu$ ($=1$). Moreover, for $\lambda$ close $\mu$, the capacity is almost zero. This is because very high $\lambda$ leads to large waiting times for qubits and thus results in significant decoherence. Furthermore, the optimal $\lambda$ depends on $\kappa$ and hence, on the physical parameters of the buffer. The idealized i.i.d. setting fails to capture this crucial dependence.

\subsection{Optimal queuing distributions}
\label{sec:queueOpt}
The effective capacity in the presence of buffer decoherence is a function of
the stationary distribution of waiting times. Thus, in turn, it is heavily
influenced by the time between preparation of two qubits and the time to
process (transmit and receive) a qubit. A quantitative understanding of this
dependence is useful for designing quantum communication systems. 

In this section, we take a short stride in that direction by characterizing the
optimal distributions in two queuing settings of general interest when the
channel and buffer decoherence follows the exponential model in Eq.
\ref{eq:expDecoh}.  The exponential decoherence model is physically the most
well motivated model for capturing decoherence in terms of the interaction time
with the environment. 

 
First, we obtain a simpler expression of the capacity result in
Theorem~\ref{thm:qcAchieve} for the exponential decoherence model. 
\begin{corollary}
\label{cor:capSeries}
The effective capacity in the presence of buffer decoherence is given by
\[\frac{\lambda}{\ln{2}} \sum_{k=1}^\infty \frac{1}{2k~(2k-1)}  \mathbb{E}_{W\sim \pi}\left[\exp\left(-\kappa~k~W\right)\right],\]
when $p_{\mbox{\em eff}}(W)~=~1~-~\exp(-\kappa~W)$ for some $\kappa>0$.
\end{corollary}
\begin{proof}
For the exponential decoherence model, the capacity expression in Theorem~\ref{thm:qcAchieve} becomes
\[\lambda~\EX_{\pi} \left[1~-~h\left(\frac{1-\exp\left(-\frac{1}{2}\kappa~ W\right)}{2}\right)\right].\]
The rest follows using the series expansion of $\log(1+x)$ for $|x|<1$ and algebraic manipulations.
\end{proof}
Note that the expression in Cor.~\ref{cor:capSeries} is valid for any stable queue, irrespective of the queuing discipline and distributions. 
 
In the queuing literature, M/G/1 and G/M/1 are two popular classes of queuing
models. In our setting, M/G/1 is equivalent to exponentially distributed
(memoryless) preparation times and generally distributed processing or service
times of qubits. G/M/1 is equivalent to generally distributed preparation times
and exponentially distributed processing or service times. As a first step
towards optimizing queuing distributions, one may ask: what are the best
distribution for processing times and preparation times in M/G/1 and G/M/1
queues, respectively? The following theorems answer this question.

\begin{theorem}
\label{thm:MD1best}
Among all quantum communication systems  with M/G/1 buffering, symmetric GAD
    channel, and exponential decoherence, the system with deterministic
    processing or service time has the maximum effective capacity for any
    $\lambda$ and $\mu$ ($>\lambda$).
\end{theorem}
\begin{proof}
Suppose there exists a service distribution for which $\mathbb{E}_{W\sim
    \pi}\left[\exp(-s W)\right]$ is more than any other service distribution
    with the same mean for any $s>0$. Then, from the capacity expression in
    Corollary~\ref{cor:capSeries}, it is clear that under that particular
    distribution, each term in the series will be more the corresponding term
    for any other distribution. Hence, that distribution will achieve the
    maximum capacity among the class of all service distributions with the same
    mean.

Thus, to complete this proof, we need only to show that for exponentially
    distributed preparation times, the deterministic service time maximizes
    $\mathbb{E}_{W\sim \pi}\left[\exp(-s W)\right]$ for any $s>0$. This follows
    directly from the proof of Theorem~4 in \cite{MandayamJagannathanEA20}.
\end{proof}

\begin{theorem}
\label{thm:DM1best}
Among all quantum communication systems  with G/M/1 buffering, symmetric GAD
    channel, and exponential decoherence, the system with deterministic
    preparation/arrival time has the maximum effective capacity for any
    $\lambda$ and $\mu$ ($>\lambda$).
\end{theorem}
\begin{proof}
Using the argument in the proof of Theorem~\ref{thm:MD1best}, it is sufficient
    to show that for exponentially distributed service times, the deterministic
    preparation/arrival time maximizes $\mathbb{E}_{W\sim \pi}\left[\exp(-s
    W)\right]$ for any $s>0$.

The following two lemmas complete the proof of this theorem.

\begin{lemma}
\label{lem:smallestSigmaBest}
Among all arrival/preparation distributions with mean $\lambda^{-1}$ (>$\mu^{-1}$), $\mathbb{E}_{W\sim \pi}\left[\exp(-s W)\right]$ for any $s>0$ is maximized by that arrival/preparation distribution for which the solution to the G/M/1 fixed point equation
\[\sigma = \mathbb{E}_{A}\left[\exp\left(-(\mu-\mu~\sigma)~A\right) \right] \]
is the smallest. 
\end{lemma}

\begin{lemma}
\label{lem:smallestSigmaDet}
Among all arrival/preparation distributions with mean $\lambda^{-1}$ (>$\mu^{-1}$), the solution to the G/M/1 fixed point equation 
\[\sigma = \mathbb{E}_{A}\left[\exp\left(-(\mu-\mu~\sigma)~A\right) \right] \]
is the smallest for the deterministic arrival/preparation time $\lambda^{-1}$.
\end{lemma}
\end{proof}

\begin{proof}[Proof of Lemma~\ref{lem:smallestSigmaBest}]
The waiting time in a G/M/1 queue is exponentially distributed with mean $\frac{1}{\mu(1-\sigma)}$, where 
$\sigma$ is the solution to the fixed point equation
\[\sigma = \mathbb{E}_{A}\left[\exp\left(-(\mu-\mu~\sigma)~A\right) \right].\]
For exponentially distributed $W$, $\mathbb{E}\left[\exp(-s W)\right]$ decreases with $\mathbb{E}[W]$. Hence, for a given $\mu$, $\mathbb{E}\left[\exp(-s W)\right]$ increases as $\sigma$ decreases, which, in turn, implies Lemma~\ref{lem:smallestSigmaBest}.
\end{proof}

Proof of Lemma~\ref{lem:smallestSigmaDet} is similar to the proof of Proposition~2 in \cite{ChatterjeeSeoEA17}.

\section{Discussion and Conclusion}
\label{sec:cncl}

\xb
\outl{Capacity fundamental but unresolved, Sec. 3:GADC $\AC_{p,n}$, for $n=1/2$
Theorem~1, no entanglement required to achieve capacity $\chi$. Additivity
means no entanglement in encoding, but decoding unknown for most additive
channels. Possibility that when additive $\chi^{(1)}(\AC_{p,n})$ no
entanglement required and $C_{Shan} = \chi$ or entanglement required and
$C_{Shan} < \chi$ also open. Additivity itself is open. Solving open i.i.d.
problems may give insights for non-i.i.d..}
\xa

Understanding the classical capacity of a quantum channel and the means by
which it can be achieved are fundamental, long-standing open issues in
quantum information.  In Sec.~\ref{SAmpDamp}, we studied these issues for the
generalized amplitude damping channel~(GADC) $\AC_{p,n}$, whose two parameters
$p$ and $n$ represent the amount of damping and mixing, respectively. In
Theorem~\ref{thm:CapEqInd}, we found that for all $p$ and $n=1/2$, the Shannon
capacity $C_{\rm Shan}$ of the (GADC) equals both its product state capacity
$\chi^{(1)}$ and classical capacity $\chi$.
In general, entanglement is required to achieve a channel's classical capacity
$\chi$. Interestingly, and of practical importance, our result implies that for the $n=1/2$ GADC, this
capacity can be achieved without using any entanglement for encoding or
decoding classical information into the channel.
Entangled encoding is not required when $\chi^{(1)}$ is additive.  Additivity
is known to hold for some values of $p,n$ of the GADC, and also known for a
variety of other channels. 

Except for the $n=1/2$ GADC solved here, for most channels with additive
$\chi^{(1)}$, finding the precise decoding that achieves capacity and
certifying whether the decoding requires entanglement or not, remain open
problems.
For our solution of the $n=1/2$ GADC, in Sec.~\ref{SsubICAmpGln}, we constructed
several families of product encoding and decoding, including one that achieves
capacity.
Our solution opens an interesting possibility. Using a product encoding and
product decoding more advanced than the ones discussed in
Sec.~\ref{SsubICAmpGln} it may be possible to find an induced channel $M_2'$,
with capacity $C(M_2') = C_{\rm Shan}(\AC_{p,n}) = \chi^{(1)}(\AC_{p,n})$ for
all $p,n$ where $\chi^{(1)}$ is additive. On the other hand, it could happen
that even when $\chi^{(1)}(\AC_{p,n})$ is additive and no entanglement is
required at the encoder, one requires entanglement at the decoder and $C_{\rm
Shan}(\AC_{p,n}) < \chi^{(1)}(\AC_{p,n})$.  Which of these possibilities is
true remains an open-problem.
To completely resolve these open problems, one would first need to find all
$p,n$ where $\chi^{(1)}(\AC_{p,n})$ is additive, a problem that also remains
open.

Insights obtained from pursuing these open problems have the potential to not
only enrich the i.i.d setting with point-to-point quantum channels but also
provide a path to study non-i.i.d queue channel settings that arise in quantum
networks. Another challenging avenue for future work is to characterise the
queue channel capacity when the underlying noise model is not additive, as
could be the case for certain parameter ranges of the GADC. This may require a
fundamentally new approach to study quantum communication networks.

\section*{Acknowledgments}
VS gratefully acknowledges support from NSF CAREER Award CCF 1652560 and NSF
grant PHY 1915407. 
The work of AC was supported by the Department of Science and
Technology, Government of India under Grant SERB/SRG/2019/001809 and Grant
INSPIRE/04/2016/001171.
PM and KJ acknowledge the Metro Area Quantum Access Network (MAQAN) project,
supported by the Ministry of Electronics and Information Technology, India vide
sanction number 13(33)/2020-CC\&BT.

\end{document}